\documentclass[submission,copyright,creativecommons]{eptcs}
\providecommand{\event}{QPL 2019}
\usepackage{underscore}
\usepackage{import}
\usepackage{latexsym}
\usepackage{amssymb}
\usepackage{amsmath}
\usepackage{amsthm}
\usepackage{stmaryrd}
\usepackage{enumerate}
\usepackage{url}
\usepackage{xypic}
\usepackage{tikz}
\usetikzlibrary{3d}
\usepackage{tikz-cd}
\usepackage[T1]{fontenc}
\usepackage{makeidx}
\usepackage{thm-restate}
\usepackage{breakurl}

\title{Scott Continuity in Generalized Probabilistic Theories}
\author{Robert Furber
\institute{Aalborg University, Denmark}
\email{furber@cs.aau.dk}
}
\def\titlerunning{Scott Continuity in Generalized Probabilistic Theories}
\def\authorrunning{Robert Furber}
\newif\ifignore

\ignoretrue

\newcommand{\catl}[1]{\ensuremath{\mathcal{#1}}}

\newcommand{\Ban}{\cat{Ban}}
\newcommand{\EffS}{\Eff_{\pm}}
\newcommand{\QED}{\hspace*{\fill}\QEDbox}
\newcommand{\Banone}{\Ban_1}
\newcommand{\absco}{\ensuremath{\mathrm{absco}}}
\newcommand{\Ball}{\ensuremath{\mathrm{Ball}}}
\newcommand{\BNS}{\cat{BNS}}
\newcommand{\cat}[1]{\ensuremath{\mathbf{#1}}}
\newcommand{\DM}{\ensuremath{\mathcal{DM}}}
\newcommand{\blank}{\ensuremath{{\mbox{-}}}}
\newcommand{\op}{\ensuremath{^{\mathrm{op}}}}
\newcommand{\Conv}{\cat{Conv}}
\newcommand{\R}{\ensuremath{\mathbb{R}}}
\newcommand{\co}{\ensuremath{\mathrm{co}}}
\newcommand{\D}{\ensuremath{\mathcal{D}}}
\newcommand{\C}{\ensuremath{\mathbb{C}}}
\newcommand{\N}{\ensuremath{\mathbb{N}}}
\newcommand{\Rgeq}{\ensuremath{\mathbb{R}_{\geq 0}}}
\newcommand{\OUS}{\cat{OUS}}
\newcommand{\Eff}{\mathcal{E}}
\newcommand{\NS}{\mathcal{NS}}
\newcommand{\Hil}{\ensuremath{\mathcal{H}}}
\newcommand{\powerset}{\ensuremath{\mathcal{P}}}
\newcommand{\ev}{\mathrm{ev}}
\newcommand{\ceil}[1]{\lceil #1 \rceil}
\newcommand{\EA}{\cat{EA}}
\newcommand{\Stat}{\ensuremath{\mathcal{S}}}
\newcommand{\EM}{\ensuremath{\mathcal{E\!M}}}
\newcommand{\QEDbox}{\ensuremath{\square}}

\newcommand{\isoarrow}{\xrightarrow{\sim}}

\newtheorem{theorem}{Theorem}[section]
\newtheorem{proposition}[theorem]{Proposition}
\newtheorem{counterexample}[theorem]{Counterexample}

\newenvironment{myproof}[1][\textnormal{\emph{Proof.}}]%
    { \begin{trivlist}%
        \item[\hskip \labelsep {\bfseries #1}]%
    }%
    { \end{trivlist}%
    }

\newcommand{\BN}{\mathcal{BN}}
\newcommand{\OU}{\mathcal{OU}}
\newcommand{\lev}{\mathrm{lev}}
\newcommand{\sdcEMod}{\cat{sdcEMod}}

\begin{document}
\maketitle

\begin{abstract}
Scott continuity is a concept from domain theory that had an unexpected previous life in the theory of von Neumann algebras. Scott-continuous states are known as normal states, and normal states are exactly the states coming from density matrices. Given this, and the usefulness of Scott continuity in domain theory, it is natural to ask whether this carries over to generalized probabilistic theories. We show that the answer is no - there are infinite-dimensional convex sets for which the set of Scott-continuous states on the corresponding set of 2-valued POVMs does not recover the original convex set, but is strictly larger. This shows the necessity of the use of topologies for state-effect duality in the general case, rather than purely order theoretic notions.
\end{abstract}

\section{Introduction}
In probability theory, the space of probability distributions forms a convex set, and in quantum theory, the space of mixed states forms a convex set. One of the main ideas of generalized probabilistic theories \cite{hardy2007, barrett07, barnum2010} is to consider a convex set $X$ as forming the space of states of a physical system. Given two states $x,y \in X$ and a real number $\alpha \in [0,1]$, we interpret the convex combination $\alpha x + (1-\alpha) y$ as the mixed state we get if we prepare $x$ with probability $\alpha$ and $y$ with probability $(1-\alpha)$. We can use this idea to define a state space abstractly in terms of such an operation, and many people have done this in slightly different ways, such as \cite{Stone1949,Neumann1970,swirszcz75}. In this introduction, we start by using $\D$-algebras as our definition of convex sets, \emph{i.e.} a convex set will be taken to be a pair $(X,\alpha_X)$ where $\alpha_X \colon \D(X) \rightarrow X$ is a function from the set of finite probability distributions $\D(X)$ to $X$, subject to some axioms (see \cite[\S 2]{furber2019} for this and other definitions). This is more general than the approach of \cite{Stone1949} but less general than \cite[Theorem 2.1]{Gudder1973} and \cite[Theorem 1]{ozawa80}. We write $\Conv = \EM(\D)$, where $\EM$ refers to the category of Eilenberg-Moore algebras of a monad (for this specifically, and other categorical notions, we refer the reader to \cite[VI.2, VI.8]{maclane}). We also mention at this point that if $\catl{C}$ is a category, and $X,Y \in \catl{C}$ are objects, then $\catl{C}(X,Y)$ refers to the set of morphisms $X \rightarrow Y$. Likewise, $\catl{C}\op$ is the category $\catl{C}$ with its arrows reversed, so $\catl{C}\op(Y,X) = \catl{C}(X,Y)$. Specific categories of structures and their homomorphisms are usually written in bold, like $\Conv$.

However we formalize abstract convex sets, 2-valued measurements are represented as affine functions to $\D(2)$, and $\D(2)$ is isomorphic to $[0,1]$, and this is in fact the more usual definition. For $X$ a convex set, we write $\Eff(X)$ for the set of affine functions $X \rightarrow [0,1]$, which is called the \emph{effects}. The effects form a structure called an \emph{effect algebra}, and $[0,1]$ is also an effect algebra, and the set of effect algebra morphisms $A \rightarrow [0,1]$, for an effect algebra $A$, forms the convex set of \emph{states}, $\Stat(A)$. We write $\EA$ for the category of effect algebras. In \cite[Theorem 17]{Jacobs10e} Bart Jacobs formulated the duality between states and effects in terms of a dual adjunction, where for any convex set $X$ and effect algebra $B$, $\EA(B, \Eff(X)) = \EA\op(\Eff(X),B) = \Conv(X,\Stat(B))$. This isomorphism arises from the evaluation morphisms $\eta_X \colon X \rightarrow \Stat(\Eff(X))$ and $\epsilon_B \colon B \rightarrow \Eff(\Stat(B))$, which are the unit and counit of the above adjunction (see \cite[IV.1 Theorems 1 and 2]{maclane} for how different definitions of adjunction relate to each other). In \cite[Theorem 4.4]{furber2019} we showed that $\eta_X$ was an isomorphism iff $X$ was isomorphic to the base of a reflexive base-norm space. Since this was not general enough to include the convex set of density matrices $\DM(\Hil)$ for $\Hil$ an infinite-dimensional Hilbert space, we altered the definition of $\Eff(X)$ to include a compact, locally convex topology, and then this was enough to make it so that $\eta_X$ was an isomorphism iff $X$ was isomorphic to the base of a Banach base-norm space \cite[Theorem 5.1]{furber2019}, which includes any closed bounded subset of a Banach space \cite[Proposition 2.4.13]{furberthesis}, such as $\DM(\Hil)$. 

There is another approach to including $\DM(\Hil)$ in such a duality that comes from \cite[Theorem 4.1]{rennelaMSc}. This is based on \emph{Scott continuity}. In domain theory \cite{abramsky95,plotkinpisa}, order-preserving maps between partially ordered sets that also preserve the least upper bounds of directed sets are known as \emph{Scott-continuous} functions. In the context of computer science, the purpose of Scott-continuous maps is to describe those functions that preserve approximation by ``less defined'' elements, which is something that computable functions must do. 

Scott-continuous maps also occur in the theory of von Neumann algebras and W$^*$-algebras. If $\Hil$ is an infinite-dimensional Hilbert space, there are states on $B(\Hil)$ that do not come from density matrices. This is unfortunate, because no physical interpretation is known for such states, seeing as they are not related back to $\Hil$. However, Scott-continuous states, which are known as the \emph{normal states}, do not have this problem, \emph{i.e.} all normal states on $B(\Hil)$ arise from density matrices. 

This notion generalizes to von Neumann algebras and the abstract version of von Neumann algebras, W$^*$-algebras\footnote{Although some authors use von Neumann algebra and W$^*$-algebra synonymously, we follow \cite{Sakai71,takesaki} in using von Neumann algebra for a concrete *-subalgebra of $B(\Hil)$ that is weakly closed (equivalently, its own bicommutant) and W$^*$-algebra to refer to a C$^*$-algebra that is isomorphic to a von Neumann algebra.}. For example, the space of bounded $\C$-valued sequences $\ell^\infty$, made into a C$^*$-algebra with pointwise operations from $\C$, is a W$^*$-algebra. The states of $\ell^\infty$ correspond to finitely-additive probability measures on $(\N,\powerset(\N))$, whereas the normal states are those states that come from a sequence of nonnegative reals $(\phi_i)_{i \in \N}$ such that $\sum_{i=1}^\infty \phi_i = 1$ under the definition $(a_i)_{i \in \N} \mapsto \sum_{i=1}^\infty a_i\phi_i$. 

In fact there is a characterization of W$^*$-algebras as C$^*$-algebras that are bounded directed-complete and are separated by their normal (\emph{i.e.} Scott-continuous) states, due to Kadison \cite[Definition 1]{kadison56}. However, it appears that the more commonly used characterization of having a predual is more easily verified in practice, so has come to predominate. 

In the approach of \cite[Appendix C]{rennelaMSc}\footnote{We slightly adapt this, using states instead of substates and effect modules instead of generalized effect modules. The approach using substates and subunital morphisms is important in practice because of the need for a quantum analogue of partial functions, but we have opted to simplify the mathematics a little in this article.}, for a convex set, the effect algebra $\Eff(X)$ is directed complete, \emph{i.e.} each directed set has a least upper bound. For directed-complete convex effect algebras $A$, we can consider $\NS(A)$ to be the subset of $\Stat(A)$ of Scott-continuous, or normal, states. Then $\Eff(X)$ also has the property, like a W$^*$-algebra, that it is separated by its normal states, \emph{i.e.} for each $a,b \in \Eff(X)$, such that $a \neq b$, there is a $\phi \in \NS(\Eff(X))$ such that $\phi(a) \neq \phi(b)$. We take the category $\sdcEMod$ to consist of effect modules\footnote{\emph{A.k.a.} convex effect algebras \cite{PulmannovaG98} \cite[\S 3]{JacobsM12b}.} $A$ that are directed-complete (hence the $\cat{dc}$) and separated by their normal states (hence the $\cat{s}$) and with morphisms that are Scott-continuous effect module homomorphisms. Then $\Eff \colon \Conv \rightarrow \sdcEMod\op$, and $\NS \colon \sdcEMod\op \rightarrow \Conv$, and $\eta_X \colon X \rightarrow \NS(\Eff(X))$ is an isomorphism if $X$ is the normal state space of a W$^*$-algebra, and $\epsilon_A \colon A \rightarrow \NS(\Eff(A))$ is an isomorphism if $A = [0,1]_B$ for $B$ a W$^*$-algebra. 

It is natural to ask if this can give an alternative formulation of \cite[Theorem 5.1]{furber2019} without topology. Concepts in domain theory, such as continuous lattices, can be formulated topologically \cite{scott72,ershov1973}, but the order-theoretic formulations are often preferred. However, we show that $\eta_X$ is not necessarily an isomorphism even if $X$ is the base of a Banach base-norm space, and $\epsilon_A$ is not necessarily an isomorphism even if $A$ is the unit interval of a separated directed-complete Banach order-unit space. 

Although our motivation comes from abstract convex sets and effect algebras, we will be using base-norm spaces and order-unit spaces throughout so that we can take advantage of the power of Banach space theory in the proofs. Therefore we give a short explanation, with pictures, of how a convex set $X$ gives rise to a base-norm space, and how the effect algebra $\Eff(X)$ sits in the dual order-unit space. 

For every closed bounded convex subset $X$ of a normed vector space, there is ``free base-norm space'' containing $X$ as its base. If we take $X$ to be the unit disc in $\R^2$, then the base-norm space we obtain is depicted in Figure \ref{BaseDiscFig}. A base-norm space is a triple $(E,E_+,\tau)$, where $E$ is a real vector space, $E_+ \subseteq E$ a cone such that $E = E_+ - E_+$ and $\tau \colon E \rightarrow \R$ a positive linear map, subject to some axioms that we will describe later. The set of $x \in E_+$ such that $\tau(x) = 1$ is called the \emph{base}, notated $B(E)$. In this case $E = \R^3$, the base is the circle drawn with a black line, and we have used four diagonal lines to indicate roughly where $E_+$ is, as it consists of nonnegative multiples of $X$. Every base-norm space has an intrinsic norm, and the unit ball of the norm is indicated using dotted lines. It is the convex hull of $B(E) \cup - B(E)$. The trace is simply the $y$-coordinate, and we will adhere to this convention in all the following figures.

\begin{figure}
\begin{minipage}{0.5\textwidth}
\caption{A base norm space}
\label{BaseDiscFig}
\begin{center}
\begin{tikzpicture}
\draw[lightgray] (0,-2,0) -- (0,2.5,0);
\draw[lightgray] (-2,0,0) -- (2,0,0);
\draw[lightgray] (0,0,-3) -- (0,0,3);
\begin{scope}[canvas is zx plane at y=1]
\draw (0,0) circle[radius=0.8cm];
\end{scope}
\draw (0,0,0) -- (1.6,2,0);
\draw (0,0,0) -- (-1.6,2,0);
\draw (0,0,0) -- (0,3,2.4);
\draw (0,0,0) -- (0,2,-1.6);
\begin{scope}[canvas is zx plane at y=-1]
\draw[dotted] (0,0) circle[radius=0.8cm];
\end{scope}
\draw[dotted] (-0.8,-1,0) -- (-0.8,1,0);
\draw[dotted] (0.8,-1,0) -- (0.8,1,0);
\end{tikzpicture}
\end{center}
\end{minipage}
\begin{minipage}{0.5\textwidth}
\caption{$\ell^1(3)$}
\label{EllOneThreeFig}
\begin{center}
\begin{tikzpicture}
\draw[lightgray] (0,-2,0) -- (0,2.5,0);
\draw[lightgray] (-2,0,0) -- (2,0,0);
\draw[lightgray] (0,0,-3) -- (0,0,3);
\draw (-1,1.1,0.58) -- (1,1.1,0.58) -- (0,1.1,-1.16) -- (-1,1.1,0.58);
\draw (0,0,0) -- (2.72,3,1.58);
\draw (0,0,0) -- (-1.82,2,1.05);
\draw (0,0,0) -- (0,2,-2.11);
\draw[dotted] (1,-1.1,-0.58) -- (-1,-1.1,-0.58) -- (0,-1.1,1.16) -- (1,-1.1,-0.58);
\draw[dotted] (-1,1.1,0.58) -- (-1,-1.1,-0.58);
\draw[dotted] (-1,1.1,0.58) -- (0,-1.1,1.16);
\draw[dotted] (1,1.1,0.58) -- (0,-1.1,1.16);
\draw[dotted] (1,1.1,0.58) -- (1,-1.1,-0.58);
\draw[dotted] (0,1.1,-1.16) -- (-1,-1.1,-0.58);
\draw[dotted] (0,1.1,-1.16) -- (1,-1.1,-0.58);
\end{tikzpicture}
\end{center}
\end{minipage}
\end{figure}

If instead of a disc, we take $X$ to be a triangle in $\R^2$, we get Figure \ref{EllOneThreeFig}. There is a clear difference with the unit ball of the intrinsic norm -- it is not a triangular prism because the triangle $X$ and its negation are different in $\R^2$. This is because $X$ is not absolutely convex. As we have indicated in the caption, Figure 2 is actually a well-known space, $\R^3$ equipped with the pointwise ordering and the trace $\tau(x,y,z) = x + y + z$. This space is known as $\ell^1(3)$. 

Of course, normally when drawing $\ell^1(3)$, we would take the three dark lines to be the coordinate axes, but the way it is depicted in Figure \ref{EllOneThreeFig} adheres to the previous convention that the trace should be the $y$ coordinate, and so is easier to compare to Figure \ref{BaseDiscFig}. 

The dual space\footnote{The space of continuous linear maps to $\R$.} $E^*$ of a base-norm space $E$ is an order-unit space. An \emph{order-unit space} is a triple $(A,A_+,u)$ where $A$ is a real-vector space, $A_+$ a positive cone such that $A = A_+ - A_+$ and $u \in A_+$ an element subject to some axioms that make the order-interval $[-u,u]$ the unit ball of an intrinsically-defined norm on $A$. The important thing is that the dual Banach space of a base-norm space $E$, \emph{i.e.} the space $E^*$ of bounded linear maps $E \rightarrow \R$ is an order-unit space and vice versa. The unit interval $[0,1]_A$, defined to be $[0,u]$ in the order defined by $A_+$, is a convex effect algebra\footnote{\emph{A.k.a.} effect module.}, and $[0,1]_{E^*}$ is exactly $\Eff(B(E))$, the effects of $B(E)$. 

We depict the order unit space $(A,A_+,u)$ dual to the base-norm space from Figure \ref{BaseDiscFig} in Figure \ref{OUDiscFig}. The cone involved is self-dual, so looks the same again. We have outlined both $A_+$ and $u-A_+$, the set of elements $\leq u$ with diagonal lines. We see that the boundary of $u - A_+$ meets the boundary of $A_+$ in a circle, which we have drawn as a dotted line. The unit ball, $[-1,1]_A$, is just $2[0,1]_A - 1$, so we do not bother to draw it.
\begin{figure}
\begin{minipage}{0.5\textwidth}
\caption{The dual of Figure \ref{BaseDiscFig}}
\label{OUDiscFig}
\begin{center}
\begin{tikzpicture}
\draw[lightgray] (0,-2,0) -- (0,2.5,0);
\draw[lightgray] (-2,0,0) -- (2,0,0);
\draw[lightgray] (0,0,-3) -- (0,0,3);
\begin{scope}[canvas is zx plane at y=1]
\draw[densely dotted] (0,0) circle[radius=0.8cm];
\end{scope}
\draw (0,0,0) -- (2,2.5,0);
\draw (0,0,0) -- (-2,2.5,0);
\draw (0,2,0) -- (-2,-0.5,0);
\draw (0,2,0) -- (2,-0.5,0);
\end{tikzpicture}
\end{center}
\end{minipage}
\begin{minipage}{0.5\textwidth}
\caption{$\ell^\infty(3)$}
\label{EllInftyThreeFig}
\begin{center}
\begin{tikzpicture}
\draw[lightgray] (0,-1,0) -- (0,3.5,0);
\draw[lightgray] (-2,0,0) -- (2,0,0);
\draw[lightgray] (0,0,-3) -- (0,0,3);
\draw (0,0,0) -- (1.22,1,-0.71);
\draw (0,0,0) -- (-1.22,1,-0.71);
\draw (0,0,0) -- (0,1,1.41);
\draw[densely dotted] (0,1,1.41) -- (1.22,2,0.71);
\draw[densely dotted] (0,1,1.41) -- (-1.22,2,0.71);

\draw[densely dotted] (1.22,1,-0.71) -- (1.22,2,0.71);
\draw[densely dotted] (1.22,1,-0.71) -- (0,2,-1.41);

\draw[densely dotted] (-1.22,1,-0.71) -- (-1.22,2,0.71);
\draw[densely dotted] (-1.22,1,-0.71) -- (0,2,-1.41);
\draw (0,2,-1.41) -- (0,3,0);
\draw (1.22,2,0.71) -- (0,3,0);
\draw (-1.22,2,0.71) -- (0,3,0);
\end{tikzpicture}
\end{center}
\end{minipage}
\end{figure}

For the dual of $\ell^1(3)$, which is $\ell^\infty(3)$, we do not draw the lines outlining the positive cone and the down set of $1$ because then the diagram, Figure \ref{EllInftyThreeFig}, would then be too cluttered. Similarly to Figure \ref{OUDiscFig}, we use dotted lines to outline where the boundaries of the positive cone and the down set of $1$ intersect. This time it is a zig-zag that does not lie in one plane. As in the case of $\ell^1(3)$, we would more usually draw the cube $[0,1]_{\ell^\infty(3)}$ as a cube in the positive orthant of $\R^3$, with the unit being $(1,1,1)$, but we have drawn it with the unit on the vertical axis so as to make the connection to Figure \ref{OUDiscFig} apparent.

In the article, we give an explicit description of the free base-norm space on the unit ball of a normed space $E$ in terms of $E$ and its norm. We call this space $\BN(E)$. We then show how to construct an order-unit space out of a normed space $E$, which we call $\OU(E)$, and show that $\BN(E)^* \cong \OU(E^*)$ and $\OU(E) \cong \BN(E^*)$. The underlying order-unit spaces of the Jordan algebras known as \emph{spin factors}\footnote{The finite dimensional cases being called $\mathfrak{S}_N$ in \cite{jordan}, the infinite-dimensional case originating the work of Topping\cite{topping1966}, for a textbook treatment, see \cite[Chapter 6]{olsen}.} are $\OU(\Hil)$ for $\Hil$ a real Hilbert space. In particular, if $\Hil$ is 3-dimensional, $\OU(\Hil)$ is isomorphic to the order-unit space of self-adjoint elements of the C$^*$-algebra of $2\times 2$ matrices, by the usual construction of the Bloch sphere. State spaces that are expressible as the unit balls of normed spaces have come up naturally when violating Tsirelson's bound, for example the ``square bit'' from boxworld \cite{gross2010} is the unit ball of $\ell^\infty(2)$, and the unit balls of $\ell^p(2)$ have also been considered as state spaces in \cite[Figure 1]{versteeg09} as for theories interpolating between ordinary quantum theory and boxworld. 

It follows from facts that are true for all base-norm spaces that the effect algebra of $\BN(E)^*$ is directed complete, and that $\BN(E)^*$ is directed complete. We then show that any bounded directed set in $\BN(E)^*$ converges in \emph{norm} to its supremum\footnote{In any infinite-dimensional W$^*$-algebra, there is a monotone sequence of projections that does not converge in norm to its supremum, so this makes spaces of the form $\BN(E)^*$ very different from W$^*$-algebras.}. Therefore every element of $\BN(E)^{**}$ is Scott-continuous. It then follows from the structure of the isomorphism $\BN(E)^{**} \cong \BN(E^{**})$, that for any \emph{irreflexive} Banach space, \emph{i.e.} a Banach space such that the evaluation mapping $E \rightarrow E^{**}$ is not an isomorphism, gives a base-norm space $\BN(E)$ that cannot be recovered from the Scott-continuous states on its dual space. We can also obtain order-unit spaces of the form $\OU(E)$ with more than one isometric predual, showing that uniqueness of preduals also fails for order-unit spaces that do not come from W$^*$-algebras.

\section{Background and Basic Results}
If $(P,\leq)$ is a poset, a set $D \subseteq P$ is \emph{directed} iff for each $x, y \in D$ there exists $z \in D$ such that $x \leq z$ and $y \leq z$. We say that $P$ is \emph{directed complete} or a \emph{dcpo} if each directed set has a least upper bound, \emph{a.k.a.} a supremum. We say that $P$ is \emph{bounded directed complete} if for each directed set $D \subseteq P$ that is bounded above, in the sense that there exists an element $u \in P$ such that for all $x \in D$, $x \leq u$, there exists a least upper bound for $D$. 

In a real vector space $E$, if we are given a finite set $(x_i)_{i \in I}$ of elements of $E$ and a matching set $(\alpha_i)_{i \in I}$ of numbers in $[0,1]$ such that $\sum_{i \in I}\alpha_i = 1$, then the point $\sum_{i \in I}\alpha_i x_i$ is called a \emph{convex combination} of the $(x_i)_{i \in I}$. A set $X \subseteq E$ is called \emph{convex} if it is closed under convex combinations, and this is equivalent to being closed under convex combinations consisting of 2 points. The \emph{convex hull} of $X \subseteq E$, written $\co(X)$, is the smallest convex set in $E$ containing $X$. This can be constructed either by taking the union of all convex combinations from $X$, or the intersection of all convex subsets of $E$ containing $X$. 

An \emph{absolutely convex combination} of the points $(x_i)_{i \in I}$ is defined using a set $(\alpha_i)_{i \in I}$ of numbers in $\R$ such that $\sum_{i \in I}|\alpha_i| \leq 1$. A set $X \subseteq E$ is \emph{absolutely convex} if it is closed under absolutely convex combinations. It is easy to prove that a set $X$ is absolutely convex iff it is convex and $X = -X$. The \emph{absolutely convex hull} of $X \subseteq E$, $\absco(X)$ is the smallest absolutely convex set containing $X$. If $X$ is non-empty and convex, then $\absco(X) = \co(X \cup -X)$, an important fact in the theory of base-norm spaces. 

An absolutely convex set $B \subseteq E$ is called \emph{absorbent} (or \emph{absorbing}) if for all $x \in E$, there exists $\alpha \in \Rgeq$ such that $x \in \alpha B$. It follows that $x \in \beta B$ for all $\beta \geq \alpha$. The \emph{gauge} or \emph{Minkowski functional} $\|\blank\|_B$ of an absorbent absolutely convex set $B$ is defined \cite[II.1.14, p. 39]{schaefer} by:
\[
\| x \|_B = \inf \{ \alpha \in \Rgeq \mid x \in \alpha B \}.
\]

This is a seminorm. For any seminorm $\|\blank\| \colon E \rightarrow \Rgeq$, we can define the closed unit ball of the seminormed space $(E,\|\blank\|)$
\[
\Ball(E) = \{ x \in E \mid \| x \| \leq 1 \},
\]
and $\|\blank\|_{\Ball(E)} = \|\blank\|$ for all seminorms. 

We say that an absolutely convex subset $B$ of a real vector space $E$ is \emph{radially bounded} (respectively \emph{radially compact}) if for every line through the origin $L \subseteq E$, the set $L \cap B$ is bounded (respectively compact) in $L$, where boundedness or compactness is defined by choosing an isomorphism $L \cong \R$. The seminorm $\|\blank\|_B$ is a norm iff $B$ is radially bounded \cite[Lemma 0.1.5]{furberthesis}. We always have $B \subseteq \Ball(E, \|\blank\|_B)$, and if $B$ is radially compact, then $B = \Ball(E, \|\blank\|_B)$ \cite[Lemma 0.1.7]{furberthesis}. 

For any normed space, $E$, the dual space $E^*$ is the space of bounded linear functions $E \rightarrow \R$, equipped with the dual norm, whose unit ball is 
\[
\Ball(E^*) = \{ \phi \colon E \rightarrow \R \mid \forall x \in \Ball(E). |\phi(x)| \leq 1 \}.
\]
Every normed space $E$ embeds canonically in its double dual, by interpreting elements of $E$ as functions on $E^*$. We call this mapping $\ev \colon E \rightarrow E^{**}$, defined for $x \in E$ and $\phi \in E^*$ by
\[
\ev(x)(\phi) = \phi(x).
\]
It is linear, and if $E^*$ and $E^{**}$ are given their dual norms, an isometry, but not necessarily surjective \cite[II.3 Theorem 19]{dunford}.

A \emph{wedge} in a real vector space $E$ is a subset $E_+ \subseteq E$ that is closed under addition and multiplication by scalars from $\Rgeq$. A wedge defines a preorder $\leq$ on $E$ by
\[
x \leq y \Leftrightarrow y - x \in E_+.
\]
This preorder is a partial order iff the wedge is a \emph{cone}, which means that additionally $E_+ \cap -E_+ = \{0\}$. A pair $(E,E_+)$ of a real vector space and a cone is called a \emph{partially ordered vector space}. We say that $E_+$ is \emph{generating} iff $E_+$ generates $E$ as a vector space, which is equivalent to $E = E_+ - E_+$. It is also equivalent to $E$ being directed in the usual sense of order theory, \emph{i.e.} for all $x,y \in E$ there exists $z \in E$ such that $x,y \leq z$. A linear map $f \colon (E,E_+) \rightarrow (F,F_+)$ between ordered vector spaces is called \emph{positive} if $f(E_+) \subseteq F_+$. For linear maps, this is equivalent to being monotone with respect to the orders defined by the cones.

An \emph{order unit} in an ordered vector space $(E,E_+)$ is an element $u \in E_+$ such that for all $x \in E$ there exists $n \in \N$ such that $-nu \leq x \leq nu$. The existence of an order unit implies that $E_+$ is generating. We say that it is \emph{archimedean} if $x \leq \frac{1}{n}u$ for all $n \in \N$ implies $x \in -E_+$. An \emph{order-unit space} is a triple $(A,A_+,u)$ where $(A,A_+)$ is an ordered vector space and $u \in A_+$ an archimedean order unit. It has a canonical norm, defined to be the Minkowski functional of the absolutely convex set $[-u,u] \subseteq A$. Order-unit spaces form a category $\OUS$, where the morphisms are positive linear maps that preserve the unit. A \emph{state} on an order-unit space is a positive linear map $\phi \colon A \rightarrow \R$ such that $\phi(u) = 1$. The set of states $\Stat(A) = \OUS(A,\R)$ is a convex subset of $A^*$, the dual space of $A$. We can also consider Scott-continuous linear functionals, which are the positive linear maps $A \rightarrow \R$ that preserve directed suprema. We will only consider these when $A$ is bounded directed-complete. The \emph{normal states} $\NS(A)$ are the Scott-continuous states. A \emph{normal linear functional} is a $\phi \colon A \rightarrow \R$ that is in the linear span of the Scott-continuous linear functionals. It need not be Scott-continuous because it need not be positive. 

Consider a triple $(E,E_+,\tau)$ where $(E,E_+)$ is an ordered vector space and $\tau \colon E \rightarrow \R$ is a positive linear functional, where $\R$ has its usual positive cone $\R_+ = [0,\infty)$. We define the base
\[
B(E) = \{ x \in E_+ \mid \tau(x) = 1 \}
\]
and the ball $U(E)$ to be the absolutely convex hull of $B(E)$. When there is no ambiguity about the ambient space $E$, we write $B$ and $U$ for the base and its absolutely convex hull. We say that $(E,E_+,\tau)$ is a \emph{pre-base-norm space} iff either $E = 0$ or $\tau \neq 0$, and the Minkowski function defined by $U(E)$ is a norm, which is equivalent to $U(E)$ being radially bounded. If $E_+$ is closed in this norm, we say $(E,E_+,\tau)$ is a \emph{base-norm space}. This holds automatically if $U(E)$ is radially compact. We write $\BNS$ for the category of base-norm spaces, where the morphisms are positive linear maps preserving the trace. For an order-unit space $(A,A_+,u)$, the space $(A^*,A_+^*,\ev(u))$, where $A_+^*$ is the space of positive linear maps $A \rightarrow \R$, is a base-norm space, whose base is $\Stat(A)$. 

For any bounded convex subset $X$ of a normed space (in fact in more generality than this, such as a convex prestructure \cite[\S 3]{ozawa80} or a $\D$-algebra \cite[Lemma 4.2]{furber2019} or \cite[Proposition 2.4.15]{furberthesis}) we can define an order-unit space $\EffS(X)$ to consist of the bounded affine $\R$-valued functions on $X$. The vector space operations and order are defined pointwise, and the order unit is simply the constant $1$ function. The notation $\EffS(X)$ is intended to imply that it is the ``signed effects'', $\Eff(X)$ being the unit interval of $\EffS(X)$. For any base-norm space $(E,E_+,\tau)$, $\EffS(B(E))$ is an order-unit space. For each (norm) bounded linear functional $a \colon E \rightarrow \R$, we can define $\rho(a) \in \EffS(B(E))$ to be the restriction of $a$ to $B(E)$. Then the map $\rho$ is an isomorphism of order unit spaces $\EffS(B(E)) \cong (E^*,E^*_+,\tau)$ \cite[Prop. 2.4.17 and Thm. 2.4.18]{furberthesis}. 

Using this isomorphism, we can prove the following fact, the analogue in our setting of \cite[Lemma C.1]{rennelaMSc}. 

\begin{restatable}{lemma}{DCLem}
\label{DCLemma}\hfill
\begin{enumerate}[(i)]
\item Let $X$ be a $\D$-algebra, \emph{e.g.} a bounded convex subset of a Banach space. If $(a_i)_{i \in I}$ is a directed set in $\EffS(X)$, which is bounded from above, then $(a_i)_{i \in I}$ has a supremum, to which it converges pointwise. 
\item Let $(E,E_+,\tau)$ be a base-norm space. Then $(E^*,E^*_+,\tau)$ is bounded directed-complete, and bounded directed sets $(a_i)_{i \in I}$ converge to their suprema in the weak-* topology (\emph{i.e.} the $\sigma(E^*,E)$ topology).
\end{enumerate}
\end{restatable}
For reasons of space, certain proofs are in the appendix.

\begin{restatable}{proposition}{PredualFnlsNormal}
\label{PredualFnlsNormalProp}
Let $(E,E_+,\tau)$ be a base-norm space. The map $\ev \colon E \rightarrow E^{**}$ maps elements of $E$ to normal linear functionals on the bounded directed-complete order-unit space $(E^*,E^*_+,\tau)$. It follows that $E^*$ is separated by its normal linear functionals, and in fact, its normal states. 
\end{restatable}

\section{Base-Norm and Order-Unit Spaces from Normed Spaces}
In the following section, it will be helpful to have the following definitions. Given two normed spaces $E$ and $F$, we will use two different norms on $E \times F$. The first makes $E \times F$ into the \emph{$\ell^1$-direct sum} $E \oplus_1 F$ and is defined by
\[
\| (x,y) \| = \| x\| + \|y\|.
\]
The second makes $E \times F$ into the \emph{$\ell^\infty$-direct sum} $E \oplus_\infty F$ and is defined by
\[
\| (x,y) \| = \max \{ \|x\|, \|y\| \}.
\]
These definitions are standard in the theory of Banach spaces \cite[p. 5]{handbookbanach}, so we leave verifying that they are norms to the reader, if this is necessary. For perspective, we note that in the category $\Banone$, whose objects are Banach spaces and whose morphisms are contractions, \emph{i.e.} linear maps of operator norm $\leq 1$, the $\ell^\infty$-direct sum is the product \cite[Example 2.1.7.d]{Borceux94}, and the $\ell^1$-direct sum is the coproduct \cite[Example 2.2.4.h]{Borceux94}. 

The following lemma collects some standard facts for easy reference.
\begin{restatable}{lemma}{StandardDirectSum}
\label{StandardDirectSumLemma}
The projections $\pi_1 \colon E \times F \rightarrow E$ and $\pi_2 \colon E \times F \rightarrow F$ are bounded for both the $\ell^1$ and $\ell^\infty$ direct sums. A sequence $(x_i,y_i)_{i \in \N}$ converges to $(x,y)$ in $E \oplus_1 F$ iff $(x_i,y_i)_{i \in \N}$ converges to $(x,y)$ in $E \oplus_\infty F$ iff $(x_i)_{i \in \N}$ converges to $x$ in $E$ and $(y_i)_{i \in I}$ converges to $y$ in $F$. 
\end{restatable}

Let $E$ be a normed space. We write $\|\blank\|_E$ when we want to emphasize that we are talking about the norm of $E$. We describe how to build a base-norm space $\BN(E)$ such that the base $B(\BN(E)) \cong \Ball(E)$. The underlying vector space of $\BN(E)$ is $E \times \R$. We define
\[
\BN(E)_+ = \{ (x,y) \in E \times \R \mid \|x\|_E \leq y \},
\]
noting that $\|x\|_E \leq y$ can only occur if $y \geq 0$, because norms only take nonnegative values. The trace $\tau \colon \BN(E) \rightarrow \R$ is defined by $\tau(x,y) = y$. 

\begin{proposition}
\label{BNDefProp}
For any normed space $E$, $(\BN(E),\BN(E)_+,\tau)$ is a base-norm space, and is a Banach base-norm space if $E$ is a Banach space. The projection $\pi_1 \colon \BN(E) \rightarrow E$ restricts to an affine isomorphism $B(\BN(E)) \rightarrow \Ball(E)$, so $\BN(E)$ is the free base-norm space on $\Ball(E)$. The unit ball of $\BN(E)$ is $\Ball(E) \times [-1,1]$, and the norm of $\BN(E)$ can be characterized as
\[
\|(x,y)\| = \max \{ \|x\|, |y| \},
\]
so the underlying normed space of $\BN(E)$ is the $\ell^\infty$-direct sum $E \oplus_\infty \R$.
\end{proposition}
\begin{proof}
We first show that $\BN(E)_+$ is a wedge. If $(x,y) \in \BN(E)_+$, \emph{i.e.} $\|x\|_E \leq y$, then for all $\alpha \geq 0$, we have $\|\alpha x\|_E = \alpha \|x\|_E \leq \alpha y$, so $\alpha (x,y) \in \BN(E)_+$. If, for $i \in \{1,2\}$, we have $(x_i,y_i) \in \BN(E)_+$ and therefore $\|x_i\|_E \leq y_i$, we have $\|x_1 + x_2 \|_E \leq \|x_1\|_E + \|x_2\|_E \leq y_1 + y_2$, so $(x_1,y_1) + (x_2,y_2) \in \BN(E)_+$. 

To show $\BN(E)_+$ is a cone, suppose that $(x,y) \in \BN(E)_+$ and $-(x,y) \in \BN(E)_+$. This means that $y \geq 0$ and $y \leq 0$, so $y = 0$. From this it follows that $\|x\|_E \leq 0$, and so $x = 0$ by the defining property of a norm. 

To show that $\BN(E)_+$ generates $\BN(E)$, let $(x,y) \in E \times \R$. If $\|x \| \leq y$, then $(x,y) \in \BN(E)_+$. If, on the other hand, $\|x\| > y$, we have $\|x \|- y > 0$, and therefore $(0,\|x\| - y) \in \BN(E)_+$. Tautologically, we also have $(x,\|x\|) \in \BN(E)_+$. As $(x,y) = (x,\|x\|) - (0,\|x\|-y)$, we have shown that the span of $\BN(E)_+$ is $\BN(E)$. 

The map $\tau$ is linear, and it is positive because if $(x,y) \in \BN(E)_+$ then $y \geq 0$. As $0 \in E$, no matter how $E$ is defined, we have $(0,1) \in \BN(E)_+$, and $\tau(0,1) = 1$, so $\tau$ is not the zero functional. 

We now need to concern ourselves with the unit ball of $\BN(E)$. We define $B = \Ball(E)$. Observe that 
\[
B(\BN(E)) = \{ (x,y) \in E \times \R \mid \| x \| \leq y \text{ and } \tau(x,y) = 1 \} = \{ (x,1) \in E \mid \|x \| \leq 1 \} = B \times \{1\}.
\]

Once we have finished proving $\BN(E)$ is a base-norm space, this shows that $\pi_1 \colon \BN(E) \rightarrow E$ restricts to an affine isomorphism $B(\BN(E)) \cong B$. 

Observe that $-B(\BN(E)) = - (B \times \{1\}) = -B \times \{-1\} = B \times \{-1\}$ by the absolute convexity of $B$. The unit ball of $\BN(E)$ is defined to be the absolutely convex hull of $B(\BN(E))$, which, as it is nonempty, is the same as $\co(B \times \{-1\} \cup B \times \{1\})$. Since $B \times \{1\}$ is already convex, we only need to use convex combinations with two elements to produce all of $U$, so 
\[
U = \{ \alpha (x_+,1) + (1-\alpha)(x_-,-1) \mid \alpha \in [0,1], x_+,x_- \in B \}. 
\]
Now, as $B$ is convex, the point $x' = \alpha x_+ + (1-\alpha)x_- \in B$, and we have
\begin{align*}
\alpha (x', 1) + (1-\alpha)(x',-1) &= \alpha (\alpha x_+ + (1-\alpha)x_-,1) + (1-\alpha)(\alpha x_+ + (1-\alpha)x_-,-1) \\
 &= (\alpha(\alpha x_+ + (1-\alpha)x_+) + (1-\alpha)(\alpha x_- + (1-\alpha)x_-), \alpha - (1 -\alpha)) \\
 &= (\alpha x_+ + (1-\alpha)x_-, \alpha - (1-\alpha)) \\
 &= \alpha (x_+,1) + (1-\alpha)(x_-,-1).
\end{align*}
This shows that every element of $U$ is an element of $B \times [-1,1]$. The opposite inclusion follows from $B \times \{1\} \subseteq B \times [-1,1]$ and the absolute convexity of $B \times [-1,1]$. So $U = B \times [-1,1]$. It follows that
\begin{align*}
\|(x,y)\|_U &= \inf \{ \alpha \in \Rgeq \mid (x,y) \in \alpha U \} = \inf \{ \alpha \in \Rgeq \mid (x,y) \in \alpha (B \times [-1,1]) \} \\
 &= \inf \{ \alpha \in \Rgeq \mid \| x \| \leq \alpha \text{ and } | y | \leq \alpha \} = \max \{ \|x\|, |y| \}.
\end{align*}
From this it follows that $\BN(E) = E \oplus_\infty \R$ as a normed space. 

We can finally show that $U$ is radially compact. If $\|(x,y)\| =1$ then either $\|x\| = 1$ and $y \leq 1$ or $|y| = 1$ and $\|x\| \leq 1$. In either case, $(x,y) \in B \times [-1,1] = U$. It follows from this that in any ray $L \subseteq \BN(E)$, the elements of norm $1$ form a closed interval, so $U$ is radially compact. From this it follows that $\BN(E)$ is a base-norm space by \cite[Proposition 2.2.6 (ii)]{furberthesis}. We have that $\BN(E)$ is the free base-norm space on $B$ because affine isomorphisms between the bases of base-norm spaces extend to isomorphisms of base-norm spaces by \cite[Corollary 2.4.9]{furberthesis}. 

Finally, if $E$ is a Banach space, then $B$ is $\sigma$-convex, so $\BN(E)$ is a Banach base-norm space by \cite[Proposition 2.4.11]{furberthesis}. 
\end{proof}

If $\Hil$ is the two-dimensional real Hilbert space, then Figure \ref{BaseDiscFig} depicts $\BN(\Hil)$. The corresponding construction of an order unit space $\OU(E)$ has underlying space $\OU(E) = E \times \R$, and the positive cone $\OU_+(E)$ defined in the same way as for $\BN(E)$. We then define the unit $u = (0,1)$. 

\begin{restatable}{proposition}{OUDefPr}
\label{OUDefProp}
For any normed space $E$, $(\OU(E),\OU(E)_+,u)$ is an order-unit space, and is a Banach order-unit space if $E$ is a Banach space. The unit ball
\[
[-1,1]_{\OU(E)} = \{ (x,y) \in E \times \R \mid \|x\|_E + |y| \leq 1 \},
\]
so $\|(x,y)\| = \|x\|_E + |y|$, \emph{i.e.} the underlying normed space of $\OU(E)$ is the $\ell^1$-direct sum $E \oplus_1 \R$. 
\end{restatable}

To characterize the dual spaces of the previous constructions, we need a pairing. By this we mean a bilinear map $\langle \blank, \blank \rangle \colon E \times F \rightarrow \R$ between vector spaces $E$ and $F$. Since we have used $E^*$ for the continuous dual space, we write $E^\odot$ for the algebraic dual, \emph{i.e.} the vector space of all linear maps $E \rightarrow \R$. We write $\lev$ for the function we get by currying a pairing $\langle \blank, \blank \rangle \colon E \times F \rightarrow \R$ on the left, so $\lev \colon F \rightarrow E^\odot$ is defined by 
\[
\lev(y)(x) = \langle x, y \rangle.
\]
For the pairings we deal with, we will ensure that the range of $\lev$ is contained in $E^*$. 

We now characterize the dual spaces of $\BN(E)$ and $\OU(E)$ in terms of the dual space of $E$. We define a pairing which can either be interpreted as being $\langle \blank, \blank \rangle \colon \BN(E) \times \OU(E^*) \rightarrow \R$ or $\OU(E) \times \BN(E^*) \rightarrow \R$, as they have the same underlying vector spaces. Here is the definition:
\begin{equation}
\label{DualPairingDefEqn}
\langle (x,\lambda), (\phi, \mu) \rangle = \phi(x) + \lambda\mu.
\end{equation}

\begin{proposition}
\label{OUDualCharacProp}
For any normed space $E$, the pairing $\langle \blank, \blank \rangle \colon \OU(E) \times \BN(E^*) \rightarrow \R$ from \eqref{DualPairingDefEqn} is bilinear and defines an isomorphism of base-norm spaces $\lev \colon \BN(E^*) \rightarrow \OU(E)^*$.
\end{proposition}
\begin{proof}
We show that it is bilinear on the left, as the argument for the right hand side is similar. Let $(x_1,\lambda_1),(x_2,\lambda_2) \in \OU(E)$ and $(\phi,\mu) \in \BN(E^*)$. Then
\begin{align*}
\langle (x_1,\lambda_1) + (x_2, \lambda_2), (\phi, \mu) \rangle &= \langle (x_1 + x_2, \lambda_1 + \lambda_2), (\phi,\mu) \rangle = \phi(x_1 + x_2) + (\lambda_1 + \lambda_2)\mu \\
 &= \phi(x_1) + \lambda_1\mu + \phi(x_2) + \lambda_2\mu = \langle (x_1,\lambda_1), (\phi,\mu) \rangle + \langle (x_2,\lambda_2), (\phi,\mu) \rangle.
\end{align*}
If $\alpha \in \R$, $(x,\lambda) \in \OU(E)$ and $(\phi,\mu) \in \BN(E^*)$, we have
\begin{align*}
\langle \alpha (x,\lambda), (\phi,\mu) \rangle &= \langle (\alpha x, \alpha \lambda), (\phi,\mu) \rangle = \phi(\alpha x) + \alpha\lambda\mu = \alpha\phi(x) + \alpha\lambda\mu = \alpha(\phi(x) + \lambda\mu) \\
 &= \alpha \langle (x,\lambda), (\phi,\mu) \rangle.
\end{align*}
We now show that $\lev \colon \BN(E^*) \rightarrow \OU(E)^\odot$ is injective. It suffices to show that its kernel is $\{0\}$. So let $(\phi,\mu) \in \BN(E^*)$, and suppose that for all $(x,\lambda) \in \OU(E)$ we have $\lev(\phi,\mu)(x,\lambda) = 0$, \emph{i.e.} $\phi(x) + \lambda\mu = 0$. In particular, this implies that for all $(x,0) \in \OU(E)$, we have $\phi(x) = 0$, so $\phi = 0$. Then, for $(x,1) \in \OU(E)$, we have $0 = \phi(x) + \mu = \mu$. So we have shown $(\phi,\mu) = 0$.

We now show that if $(\phi,\mu) \in \BN(E^*)_+$, then $\lev(\phi,\mu)$ is positive. So let $(\phi,\mu) \in \BN(E^*)_+$ and $(x,\lambda) \in \OU(E)_+$. We have $|\phi(x)| \leq \|x\|\|\phi\| \leq \lambda\mu$, from the inequalities $\|x\| \leq \lambda$ and $\|\phi\| \leq \mu$ that come from the positivity of $(x,\lambda)$ and $(\phi,\mu)$ respectively. It follows that $-\phi(x) \leq \lambda\mu$ and therefore 
\[
0 \leq \phi(x) + \lambda\mu = \langle (x,\lambda),(\phi,\mu) \rangle = \lev(\phi,\mu)(x,\lambda),
\]
proving $\lev(\phi,\mu)$ is positive. It follows that $\lev(\phi,\mu)$ is bounded\footnote{This is true for any positive linear functional $\phi$ on an order-unit space $(A,A_+,u)$ because for all $a \in A$, $\|a\| \leq 1$ is equivalent to $-u \leq a \leq u$, and so for all $a \in \Ball(A)$, $\phi(-u) \leq \phi(a) \leq \phi(u)$ holds, and proves $|\phi(a)| \leq \phi(u)$.}. Since $\BN(E^*)$ is the linear span of $\BN(E^*)_+$, this implies that $\lev$ maps $\BN(E^*)$ into $\OU(E)^*$.

To finish the proof that $\lev$ is an isomorphism, we show that its restriction to $B(\BN(E^*))$ is a bijection onto $B(\OU(E)^*) = \Stat(\OU(E))$. As we already proved that $\lev$ is linear, this is an affine isomorphism of bases and so $\lev$ is an isomorphism of base-norm spaces by \cite[Proposition 3.1]{furber2019}. 

If $(\phi,1) \in B(\BN(E^*))$, we already know that $\lev(\phi,\mu)$ is positive, so we only need to show that it has trace 1 in $\OU(E)^*$. Taking the trace in the dual space of an order-unit space is done by evaluating at the unit, which is $(0,1) \in \OU(E)$. So $\lev(\phi,1)(0,1) = \phi(0) + 1 = 1$ shows that $\lev(\phi,1) \in B(\OU(E)^*)$. 

To finish, we only need to show that $\lev \colon B(\BN(E^*)) \rightarrow B(\OU(E)^*)$ is surjective. Let $\phi \in B(\OU(E)^*)$, \emph{i.e.} $\phi \colon \OU(E) \rightarrow \R$ is positive and $\phi(0,1) = 1$. Define $\psi \colon E \rightarrow \R$ by $\psi(x) = \phi(x,0)$. We show that $\psi$ is a linear map $E \rightarrow \R$ and $\|\psi\| \leq 1$ (so in particular, $\psi \in E^*$).

To prove $\psi$ is linear, observe that $\psi(0) = \phi(0,0) = 0$, and for all $x,y \in E$, $\alpha \in \R$:
\[
\psi(\alpha x+y) = \phi(\alpha x + y,0) = \phi(\alpha (x,0) + (y,0)) = \alpha \phi(x,0) + \phi(y,0) = \alpha \psi(x) + \psi(y).
\]
To see that $\|\psi\| \leq 1$, let $x \in \Ball(E)$. Then in $\OU(E)$, $\|(x,0)\| = \|x\| + |0| = \|x\|$ by Proposition \ref{OUDefProp}. As $\phi$ is a state, $\|\phi\| \leq 1$ \cite[Proposition 1.2.8]{furberthesis}, so $|\psi(x)| = |\phi(x,0)| \leq \|(x,0)\| = \|x\| \leq 1$. As this holds for all $x \in \Ball(E)$, we have shown $\|\psi\| \leq 1$. 

It follows that $(\psi,1) \in B(\BN(E^*))$. To finish the proof that $\lev$ is surjective, let $(x,\lambda) \in \OU(E)$ and observe $\lev(\psi,1)(x,\lambda) = \psi(x) + \lambda = \phi(x,0) + \phi(0,\lambda) = \phi(x,\lambda)$, \emph{i.e.} $\lev(\psi,1) = \phi$. 
\end{proof}

It also works the other way round, though it is slightly trickier, so we need a lemma.

\begin{restatable}{lemma}{BallExtension}
\label{BallExtensionLemma}
Let $E$ be a normed vector space, and suppose $a \colon \Ball(E) \rightarrow \R$ is a bounded affine map with $a(0) = 0$. Then $a$ admits a unique extension to a bounded linear map $\psi \colon E \rightarrow \R$. 
\end{restatable}

\begin{restatable}{proposition}{BNDualCharac}
\label{BNDualCharacProp}
For any normed space $E$, the pairing $\langle \blank, \blank \rangle \colon \BN(E) \times \OU(E^*) \rightarrow \R$ from \eqref{DualPairingDefEqn} is bilinear and defines an isomorphism of order-unit spaces $\lev \colon \OU(E^*) \rightarrow \BN(E)^*$.
\end{restatable}

Taking $\Hil$ to be the 2-dimensional real Hilbert space again, as $\Hil^* \cong \Hil$ via the inner product, Figure \ref{OUDiscFig} depicts any of the three isomorphic spaces $\BN(\Hil)^*$, $\OU(\Hil^*)$ or $\OU(\Hil)$. 

It follows from Lemma \ref{DCLemma} (ii) that for any normed space $E$, $\BN(E)^* \cong \OU(E^*)$ is bounded directed-complete. We can now prove the following fact about it.

\begin{proposition}
\label{OUDCNormProp}
If $E$ is a normed space, and $(y_i,\mu_i)_{i \in I}$ is a directed set in $\OU(E^*)$ with supremum $(y,\mu)$, then $(y_i,\mu_i) \to (y,\mu)$ in norm. Therefore every positive linear functional $\phi \colon \OU(E^*) \rightarrow \R$ is Scott-continuous. 
\end{proposition}
\begin{proof}
Let $(y_i,\mu_i)_{i \in I}$ be a directed set in $\OU(E^*)$ with supremum $(y,\mu)$. As $\OU(E^*) \cong \BN(E)^*$ (Proposition \ref{BNDualCharacProp}), Lemma \ref{DCLemma} (ii) implies that $(y_i,\mu_i) \to (y,\mu)$ in the weak-* topology. Since $(0,1) \in \BN(E)$, the mapping $\lev(0,1) \colon \OU(E^*) \rightarrow \R$ is weak-* continuous, so
\[
\mu_i = y_i(0) + \mu_i = \lev(0,1)(y_i,\mu_i) \to \lev(0,1)(y,\mu) = y(0) + \mu = \mu.
\]
So for all $\epsilon > 0$, there exists $i \in I$ such that for all $j \geq i$ we have $|\mu - \mu_j| < \epsilon$. As $(y,\mu)$ is the least upper bound, we have for all $j \in I$, $(y_j, \mu_j) \leq (y,\mu)$, so $\| y - y_j \|_{E^*} \leq \mu - \mu_j$. So if $j \geq i$, we have\footnote{Note that as $\mu - \mu_j \geq \| y - y_j \|_{E^*} \geq 0$, we have $\mu - \mu_j = | \mu - \mu_j |$.} $\| y - y_j \|_{E^*} \leq \mu - \mu_j = |\mu - \mu_j | < \epsilon$. As this holds for arbitrary $\epsilon > 0$, this shows that $y_j \to y$ in $E$. Since we already showed that $\mu_j \to \mu$ in $\R$, it follows that $(y_j,\mu_j) \to (y,\mu)$ in $E \oplus_1 \R$ by Lemma \ref{StandardDirectSumLemma}, and therefore in $\OU(E^*)$ by Proposition \ref{OUDefProp}. 

If $\phi \colon \OU(E^*) \rightarrow \R$ is positive, then $(\phi(y_i,\mu_i))_{i \in I}$ is a bounded directed set in $\R$. As $\phi$ is bounded, it is norm continuous, and so $\phi(y_i,\mu_i) \to \phi(y,\mu)$. Therefore $\phi(y,\mu)$ is the least upper bound of $(\phi(y_i,\mu_i))_{i \in I}$. As this holds for any bounded directed set $(y_i,\mu_i)_{i \in I}$ with supremum $(y,\mu)$ in $\OU(E^*)$, we have shown that $\phi$ is Scott-continuous. 
\end{proof}

\section{Counterexamples}
Given what we have proved so far, to obtain a convex set $X$ such that the evaluation mapping from $X$ to the normal states of $\Eff(X)$ is not surjective, we only need a normed space $E$ that is not reflexive. As a concrete example, we  use the space $c_0$. This space consists of sequences $(a_i)_{i \in \N}$ of real numbers, converging to $0$, and the vector space operations are defined pointwise, and the norm is defined as
\begin{equation}
\label{SupNormEqn}
\|(a_i)_{i \in \N}\| = \sup_{i \in \N} |a_i|.
\end{equation}
The dual is isomorphic to the space of summable sequences $\ell^1$, which is given the norm
\[
\|(\phi_i)_{i \in \N}\| = \sum_{i \in \N}|\phi_i|.
\]
The isomorphism comes from the pairing
\begin{equation}
\label{SequenceSpacePairingEqn}
\langle (a_i)_{i \in \N}, (\phi_i)_{i \in \N} \rangle = \sum_{i \in \N}a_i\phi_i,
\end{equation}
for which $\lev \colon \ell^1 \rightarrow c_0^*$ is a bijective isometry. The dual space of $\ell^1$ is isomorphic to $\ell^\infty$, the space of bounded sequences, which is equipped with the norm \eqref{SupNormEqn} and paired with $\ell^1$ via \eqref{SequenceSpacePairingEqn}. The fact that $c_0$ is not reflexive then follows from $c_0 \subseteq \ell^\infty$ being a proper subset, for example the constant $1$ sequence is in $\ell^\infty$ but not in $c_0$. 

\begin{counterexample}
\label{BNIrreflexiveCounterexample}
If $E$ is a normed space that is not reflexive, such as $c_0$, then
\begin{enumerate}[(i)]
\item $\BN(E)$ is a base-norm space such that there are Scott-continuous states $\BN(E)^* \rightarrow \R$ that are not of the form $\ev(x,\lambda)$ for any $(x,\lambda) \in \BN(E)$. 
\item $X = \Ball(E)$ is a $\D$-algebra such that there are Scott-continuous states $\Eff(X) \rightarrow [0,1]$ that are not of the form $\eta_X(x)$ for any $x \in X$. 
\end{enumerate}
\end{counterexample}
\begin{proof}
We deal with part (i) first. Since $E$ is not reflexive, there exists $\phi \in \Ball(E^{**})$ such that $\phi \neq \ev(x)$ for any $x \in E$, which is to say, for all $x \in E$, there exists $y \in E^*$ such that $\phi(y) \neq y(x)$. 

As $\lev \colon \OU(E^*) \isoarrow \BN(E)^*$ is an isomorphism of order-unit spaces (Proposition \ref{BNDualCharacProp}), it suffices to show that there are normal states on $\OU(E^*)$ that are not of the form $\ev(x,\lambda) \circ \lev$, where $(x,\lambda) \in \BN(E)$. 

By Proposition \ref{OUDCNormProp}, $\lev(\phi,1)$ defines a Scott-continuous state on $\OU(E^*)$ via the pairing \eqref{DualPairingDefEqn}. For any $(x,\lambda) \in \BN(E)$, there exists $y \in E^*$ such that $\phi(y) \neq y(x)$, and therefore 
\begin{align*}
(\ev(x,\lambda) \circ \lev)(y,0) &= \ev(x,\lambda)(\lev(y,0)) = \lev(y,0)(x,\lambda) = \langle (x,\lambda), (y,0) \rangle = \\
 y(x) + 0 &= y(x) \neq \phi(y) = \langle (y,0), (\phi,1) \rangle = \lev(\phi,1)(y,0).
\end{align*}
so $\lev(\phi,1) \neq \ev(x,\lambda) \circ \lev$ for any $(x,\lambda) \in \BN(E)$. Therefore $\lev(\phi,1)$ is the normal state on $\OU(E^*)$ that we seek. 

Part (ii) then follows from the isomorphism $\Ball(E) \cong B(\BN(E))$ (Proposition \ref{BNDefProp}), and the fullness of the embeddings $B \colon \BNS \rightarrow \Conv$ and $[0,1]_\blank \colon \OUS \rightarrow \EA$ \cite[Propositions 3.1 and 3.2]{furber2019}. 
\end{proof}

In the following, we will use the fact that a Banach space $E$ is reflexive iff $E^*$ is \cite[II.3 Corollary 24]{dunford}. 

\begin{restatable}{counterexample}{HypStoneCounter}
\label{OUIrreflexiveCounterexample}
If $E$ is a Banach space that is not reflexive, such as $c_0$. 
\begin{enumerate}[(i)]
\item To simplify notation, we write $A = E^*$. $\OU(A)$ is a bounded directed-complete order-unit space, separated by its normal states, such that the evaluation mapping $\OU(A) \rightarrow \EffS(\NS(\OU(A)))$ is not surjective, and therefore not an isomorphism. 
\item Taking $X = \Ball(E)$, $\Eff(X)$ is a directed-complete effect module, separated by its normal states, such that the evaluation mapping $\Eff(X) \rightarrow \Eff(\NS(\Eff(X)))$ is not surjective, and therefore not an isomorphism.
\end{enumerate}
\end{restatable}

The space $c$ consists of convergent real-valued sequences. It is made into a Banach space as a closed subspace of $\ell^\infty$, and $c_0$ is a closed subspace of $c$ of codimension 1. The convex set $\Ball(c_0)$ has no extreme points, but the constant $1$ function is an extreme point of $\Ball(c)$, so $\Ball(c_0) \not\cong \Ball(c)$ as convex sets, and so $\BN(c_0) \not\cong \BN(c)$. But $c^* \cong \ell^1$ isometrically \cite[\S IV.4.3, p.65]{banach}, and $c_0^* \cong \ell^1$ isometrically as described above, so $\BN(c_0)^* \cong \OU(c_0^*) \cong \OU(\ell^1) \cong \OU(c^*) \cong \BN(c)^*$ through order-unit space isomorphisms, so we have shown:
\begin{counterexample}
There exist base-norm spaces $E_1, E_2$ such that $E_1^* \cong E_2^*$ as order-unit spaces, but $E_1 \not\cong E_2$ isometrically. 
\end{counterexample}

This cannot happen with W$^*$-algebras -- if a W$^*$-algebra $A \cong E_1^*$ and $A \cong E_2^*$ isometrically, then $E_1 \cong E_2$ isometrically \cite[Corollary 1.13.3]{Sakai71}.

\subsubsection*{Acknowledgements}
Robert Furber has been financially supported by the Danish Council for Independent Research, Project 4181-00360.

\bibliographystyle{eptcs}
\bibliography{qpl2019}

\appendix
\section{Proofs}

\DCLem*
\begin{myproof}
\begin{enumerate}[(i)]
\item Let $(a_i)_{i \in I}$ be a directed set in $\EffS(X)$, bounded above by $b \in \EffS(X)$. Define $a \colon X \rightarrow \R$ by
\[
a(x) = \sup_{i \in I}a_i(x).
\]
The supremum exists because for each $x \in X$, $(a_i(x))_{i \in I}$ is a directed set in $\R$ bounded above by $b(x)$, and this also shows that $a \leq b$ so $a$ is a bounded function. Therefore we only need to show that it is affine to prove that $a \in \EffS(X)$. Let $x,y \in X$ and $\alpha \in [0,1]$. Then
\begin{align*}
a(\alpha x + (1-\alpha)y) &= \sup_{i \in I}a_i(\alpha x + (1 - \alpha)y)  \\
 &= \sup_{i \in I}(\alpha a_i(x) + (1 - \alpha)a_i(y)) \\
 &= \lim_{i \in I}(\alpha a_i(x) + (1 - \alpha)a_i(y)) \\
 &= \alpha \lim_{i \in I}a_i(x) + (1 - \alpha)\lim_{i \in I}a_i(y) & \text{continuity in $\R$} \\
 &= \alpha \sup_{i \in I}a_i(x) + (1 - \alpha)\sup_{i \in I}a_i(y) \\
 &= \alpha a(x) + (1 -\alpha)a(y),
\end{align*}
where we used the fact that directed sets in $\R$ converge to their suprema. Therefore $a \in \EffS(X)$. It follows that $a = \sup_{i \in I}a_i$ because the ordering is pointwise and this holds at each $x \in X$. Similarly, as directed sets converge to their suprema in $\R$, $(a_i)_{i \in I}$ converges pointwise to $a$ in $\EffS(X)$. 
\item We have that $\EffS(B(E))$ is isomorphic to the order-unit space $(E^*,E^*_+,\tau)$ dual to $E$, so $E^*$ is bounded directed-complete and directed sets $(a_i)_{i \in I}$ in $E^*$ converge pointwise on $B(E)$ to their suprema. Let $a$ be the supremum of $(a_i)_{i \in I}$. Convergence in the weak-* topology is equivalent to pointwise convergence on $E$, so we only need to show that $a_i(x) \to a(x)$ for all $x \in E$. To make things non-trivial, we assume that $B(E)$ is not empty. As $E$ is a base-norm space, $x = \alpha x_+ - \beta x_-$ for $\alpha,\beta \in \Rgeq$ and $x_+,x_- \in B(E)$. By the linearity of each $a_i$ and $a$, as well as the continuity of linear operations in $\R$, we have
\[
a_i(x) = \alpha a_i(x_+) - \beta a_i(x_-) \to \alpha a(x_+) - \beta a(x_-) = a(x),
\]
so $a_i \to a$ in the weak-* topology. \QED
\end{enumerate}
\end{myproof}

\PredualFnlsNormal*
\begin{proof}
To show this, it suffices to show that $\ev(x)$ is Scott-continuous for all $x \in B(E)$, because, if $E \neq 0$, every element $x \in E$ can be expressed as $x = \alpha x_+ - \beta x_-$ where $\alpha,\beta \in \Rgeq$ and $x_+,x_- \in B(E)$, and $\ev$ is a linear map. 

So let $x \in B(E)$ and let $(a_i)_{i \in I}$ be a bounded directed set, with supremum $a$. Then
\[
\ev(x)\left(\sup_{i \in I}a_i\right) = \left(\sup_{i \in I}a_i\right)(x) = \sup_{i \in I}a_i(x) = \sup_{i \in I}\ev(x)(a_i)
\]
by definition.

For the last statement, suppose that $a,b \in E^*$ and for all normal states $\phi \in \NS(E^*)$, we have $\phi(a) = \phi(b)$. By the above, we have for all $x \in B(E)$, $a(x) = \ev(x)(a) = \ev(x)(b) = b(x)$, so $a$ and $b$ agree on $B(E)$. As $E$ is spanned by $B(E)$, we have $a = b$. Therefore $E^*$ is separated by its normal states, and therefore by its normal linear functionals. 
\end{proof}

\StandardDirectSum*
\begin{proof}
We first show that $\pi_1 \colon E \oplus_1 F \rightarrow E$ is bounded. Let $(x,y) \in E \oplus_1 F$ such that $\|(x,y)\| \leq 1$, \emph{i.e.} $\| x\| + \|y\| \leq 1$. Then $\|\pi_1(x,y)\| = \|x\| \leq \|x\| + \|y\| \leq 1$, proving $\pi_1$ is bounded. The proof for $\pi_2$ is similar.

Now let us consider $\pi_1 \colon E \oplus_\infty F \rightarrow E$. Let $(x,y) \in E \oplus_\infty F$ such that $\|(x,y)\| \leq 1$, \emph{i.e.} $\max \{ \|x\| , \|y\| \} \leq 1$. Then $\|\pi_1(x,y)\| = \|x\| \leq \max \{ \|x\|,\|y\| \} \leq 1$, so $\pi_1$ is bounded. Again, the proof for $\pi_2$ is similar. 

Since bounded maps of normed spaces are continuous, it therefore follows that if $(x_i,y_i) \to (x,y)$ in either $E \oplus_1 F$ or $E \oplus_\infty F$, then $x_i \to x$ in $E$ and $y_i \to y$ in $F$. 

For the opposite implications, we consider $E \oplus_\infty F$ first. Suppose that $x_i \to x$ in $E$ and $y_i \to y$ in $F$. Then for each $\epsilon > 0$, we can find an $N$ such that for all $i \geq N$ both $\|x_i - x\| < \epsilon$ and $\|y_i - y\| < \epsilon$ (by taking the maximum of an $N$ for which this holds for $\|x_i - x\|$ and an $M$ for which this holds for $\|y_i - y\|$). Then for all $i \geq N$ we have
\[
\| (x_i,y_i) - (x,y) \| = \| (x_i - x, y_i - y) \| = \max \{ \|x_i - x\|, \|y_i - y\| \} < \epsilon,
\]
and since this holds for all $\epsilon > 0$, $(x_i,y_i) \to (x,y)$ in $E \oplus_\infty F$. 

The argument for $E \oplus_1 F$ is similar, except we use $\frac{\epsilon}{2}$, and then for all $i \geq N$ we have
\[
\| (x_i,y_i) - (x,y) \| = \|x_i-x\| + \|y_i-y\| < \frac{\epsilon}{2} + \frac{\epsilon}{2} = \epsilon,
\]
and as this holds for all $\epsilon > 0$, we have $(x_i,y_i) \to (x,y)$ in $E \oplus_1 F$. 
\end{proof}

\OUDefPr*
\begin{proof}
We have already proved that $\OU(E)_+$ is a cone that generates $\OU(E)$ in Proposition \ref{BNDefProp}. There is therefore a simpler way to prove that $u$ is a strong order unit \cite[Lemma A.5.1]{furberthesis}, which is to prove that for all $(x,y) \in \OU(E)_+$, there exists $n \in \N$ such that $(x,y) \leq n u$. 

So let $(x,y) \in \OU(E)_+$, \emph{i.e.} $\|x\|_E \leq y$. Define $n = \ceil{\|x\|_E + y}$. Then
\[
n u \geq (\|x\|_E + y)(0,1) = (0,\|x\|_E + y),
\]
and
\[
(0,\|x\|_E + y) - (x,y) = (-x,\|x\|_E) \in \OU(E)_+
\]
because $\|-x\|_E = \|x\|_E$. Therefore $(x,y) \leq (0,\|x\|_E + y) \leq n u$, as required. 

To show that $\OU(E)$ is archimedean, suppose that $(x,y) \leq \frac{1}{n}u$ for all $n \in \N$, \emph{i.e.} $(-x, \frac{1}{n} - y) \in \OU(E)_+$, and therefore $\|-x\| \leq \frac{1}{n} - y$ for all $n \in \N$. By the archimedean property of $\R$, this implies $\|-x\|_E \leq -y$, and so $-(x,y) \in \OU(E)_+$, \emph{i.e.} $(x,y) \leq 0$. 

So we have proved that $(\OU(E),\OU(E)_+,u)$ is an order-unit space. We now find the unit ball and the norm. The unit ball is $[-u,u]$, or $\{ (x,y) \in E \times \R \mid -(0,1) \leq (x,y) \leq (0,1) \}$.
We treat each of the two inequalities on $(x,y)$ in turn and then combine them.
\[
-(0,1) \leq (x,y) \Leftrightarrow (x,y+1) \geq 0 \Leftrightarrow \| x\|_E \leq y + 1 \Leftrightarrow -y \leq 1 - \|x\|_E
\]
and
\[
(x,y) \leq (0,1) \Leftrightarrow (-x,1-y) \geq 0 \Leftrightarrow \|-x\|_E \leq 1 - y \Leftrightarrow y \leq 1 - \|x\|_E
\]
Taken together, we have shown
\[
(x,y) \in [-u,u] \Leftrightarrow |y| \leq 1 - \|x\|_E \Leftrightarrow \|x\|_E + |y| \leq 1.
\]
Therefore $\|(x,y)\|_{[-u,u]} = \|x\|_E + |y|$, which is to say, the underlying normed space of $\OU(E)$ is the $\ell^1$-direct sum $E \oplus_1 \R$. As the projection maps $\pi_1 \colon E \times \R \rightarrow E$ and $\pi_2 \colon E \times \R \rightarrow \R$ are bounded and linear (Lemma \ref{StandardDirectSumLemma}), they are uniformly continuous, and so preserve Cauchyness of sequences. Therefore, if $E$ is a Banach space, for every Cauchy sequence $(x_i,\lambda_i)_{i \in \N}$ in $\OU(E)$, there is an $x \in E$ and a $\lambda \in \R$ such that $x_i \to x$ and $\lambda \to \lambda$, so by Lemma \ref{StandardDirectSumLemma} $(x_i,\lambda_i) \to (x,\lambda)$ in $\OU(E)$. This shows that if $E$ is a Banach space, so is $\OU(E)$. 
\end{proof}

\BallExtension*
\begin{proof}
Define $\psi(0) = 0$ and for all $x \in E$ with $x \neq 0$, $\psi(x) = \|x\| a\left(\frac{x}{\|x\|}\right)$. We first show that $\psi$ agrees with $a$ on $\Ball(E)$. This is immediate for $0$, and if $x \neq 0$, we have $0 < \|x\| \leq 1$ and $x = \|x\| \frac{x}{\|x\|} + (1-\|x\|)0$, which is a convex combination. So as $a$ is affine
\[
a(x) = \|x\| a\left(\frac{x}{\|x\|}\right) + (1 - \|x\|)a(0) = \|x\|a\left(\frac{x}{\|x\|}\right) = \psi(x).
\]

To show that $\psi$ is linear, we first show that it preserves scalar multiplication, in three cases. If $\lambda = 0$, then $\psi(\lambda x) = \psi(0) = 0 = 0\psi(x) = \lambda \psi(x)$. If $\lambda > 0$, then
\[
\psi(\lambda x) = \|\lambda x\|a\left(\frac{\lambda x}{\|\lambda x\|}\right) = \lambda \|x\| a\left(\frac{\lambda x}{\lambda\|x\|}\right) = \lambda \|x\|a\left(\frac{x}{\|x\|}\right) = \lambda \psi(x).
\]
For $\lambda < 0$, it suffices to do $\lambda = -1$ by what we have already done, so first observe that
\[
\psi(-x) = \|-x\|a\left(\frac{-x}{\|x\|}\right) = \|x\| a\left(\frac{-x}{\|x\|}\right).
\]
Now, if $x \in \Ball(E)$, we have
\[
0 = a(0) = a\left(\frac{1}{2}(-x) + \frac{1}{2}x\right) = \frac{1}{2}a(-x) + \frac{1}{2}a(x),
\]
which, by multiplying by 2, gives $a(-x) = -a(x)$. So
\[
\psi(-x) = \|x\|a\left(\frac{-x}{\|x\|}\right) = -\|x\|a\left(\frac{x}{\|x\|}\right) = -\psi(x).
\]
We prove that $\psi$ preserves addition as follows. Let $x,y \in E$. If $x=y=0$, then $\psi(x+y) = \psi(x) + \psi(y)$ has already been proved, because $0= 0+0$. So assume that at least one of $x$ and $y$ is not $0$. Take $\lambda = \frac{1}{\max\{\|x\|,\|y\|\}}$, which is defined by the previous assumption. Then $\lambda x$ and $\lambda y$ are in $\Ball(E)$. We have
\begin{align*}
\frac{1}{2}\lambda \psi(x + y) &= \psi\left(\frac{1}{2}\lambda x + \frac{1}{2}\lambda y\right) = a\left(\frac{1}{2}\lambda x + \frac{1}{2}\lambda y\right) = \frac{1}{2} a(\lambda x) + \frac{1}{2} a(\lambda y) \\
 &= \frac{1}{2}\psi(\lambda x) + \frac{1}{2} \psi(\lambda y) = \frac{1}{2}\lambda (\psi(x) + \psi(y)).
\end{align*}
Therefore $\psi(x+y) = \psi(x) + \psi(y)$, because $\lambda \neq 0$. We have therefore proved that $\psi$ is linear, and it is bounded because it agrees with $a$ on $\Ball(E)$, upon which it is bounded. 
\end{proof}

\BNDualCharac*
\begin{proof}
The proof that $\langle \blank,\blank \rangle$ is bilinear and $\lev \colon \OU(E^*) \rightarrow \BN(E)^\odot$ is injective are similar to the proofs in Proposition \ref{OUDualCharacProp} so are omitted. 

Since $\rho \colon \BN(E)^* \rightarrow \EffS(B(\BN(E)))$ is an isomorphism \cite[Lemma 4.3]{furber2019}, in order to prove that $\lev$ is an isomorphism $\OU(E^*) \rightarrow \BN(E)^*$, it suffices to show that $\rho \circ \lev \colon \OU(E^*) \rightarrow \EffS(B(\BN(E)))$ is an isomorphism of order-unit spaces. We already know it is injective. By Proposition \ref{BNDefProp}, $B(\BN(E)) = \Ball(E) \times \{1\}$, and we will define $B = \Ball(E)$ as a short name. We will also not bother to write $\rho$, and simply refer to the map $\rho \circ \lev$ as $\lev$, since it is just the restriction to $B \times \{1\}$. 

We first show that for all $(\phi,\mu) \in OU(E^*)$, $\lev(\phi,\mu) \colon B \times \{1\} \rightarrow \R$ is bounded. Given $x \in B$, we have $|\phi(x)| \leq \|\phi\|$, so
\[
|\lev(\phi,\mu)(x,1)| = |\phi(x) + \mu| \leq |\phi(x)| + |\mu| \leq \|\phi\| + |\mu|,
\]
and as this does not depend on $x$, this proves it is bounded. The fact that $\lev(\phi,\mu)$ is affine follows from the bilinearity of the pairing $\langle \blank, \blank \rangle$, so we have shown $\lev(\phi,\mu) \in \EffS(B \times \{1\})$. The fact that $\lev$ is itself linear then follows from the other part of the bilinearity of the pairing. 

We can see that $\lev$ is unital, because for all $x \in B$:
\[
\lev(0,1)(x,1) = 0(x) + 1 = 1,
\]
and the constant $1$ function is the unit of $\EffS(B \times \{1\})$. 

To prove that $\lev$ is positive, let $(\phi,\mu) \in \OU(E^*)$, \emph{i.e.} $\|\phi\| \leq \mu$. Then for each $x \in B$, $|\phi(x)| \leq \|\phi\|$, so $-\phi(x) \leq \|\phi\| \leq \mu$. It follows that $\lev(\phi,\mu)(x,1) = \phi(x) + \mu \geq 0$, proving that $\lev(\phi,\mu)$ is positive. 

We also prove that $\lev$ is an order-embedding, \emph{i.e.} $\lev(\phi,\mu) \in \BN(E)^*_+$ implies $(\phi,\mu) \in \OU(E^*)_+$. So we have that for all $x \in B$, $\phi(x) + \mu \geq 0$. Since $0 \in B$, we have $0 \leq \phi(0) + \mu = \mu$. Since $-x \in B$ iff $x \in B$ as well, we have both $\phi(x) + \mu \geq 0$ and $\phi(-x) + \mu \geq 0$ for all $x \in B$ as well, so $|\phi(x)| \leq \mu$ for all $x \in B$. By the definition of the dual norm, this proves $\|\phi\| \leq \mu$, and therefore $(\phi,\mu) \in \OU(E^*)_+$. 

We already know that $\lev$ is injective, so to finish proving that it is an isomorphism of order-unit spaces, we only need to show that it is surjective, because the positivity and unitality of the inverse of $\lev$ follow from the fact that $\lev$ is an order embedding and unital. Let $a \in \EffS(B \times \{1\})$. We define $\psi \colon \Ball(E) \rightarrow \R$ by $\psi(x) = a(x,1) = a(0,1)$. As $\psi$ is affine and preserves $0$ it extends (Lemma \ref{BallExtensionLemma}) to a unique bounded linear map $\psi \colon E \rightarrow \R$. Then $(\psi,a(0,1)) \in \OU(E^*)$, and for all $x \in B$:
\[
\lev(\psi,a(0,1))(x,1) = \psi(x) + a(0,1) = a(x,1) - a(0,1) + a(0,1) = a(x,1),
\]
so $\lev(\psi,a(0,1)) = a$, proving the surjectivity. 
\end{proof}

\HypStoneCounter*
\begin{proof}
As before, we deal with part (i) first and do (ii) as a corollary. As $E$ is a Banach space that is not reflexive, $A = E^*$ is not reflexive, so there exists some $\phi \in \Ball(A^{**})$ such that $\phi \neq \ev(x)$ for all $x \in A$, \emph{i.e.} for all $x \in A$ there exists $y \in A^*$ such that $\phi(y) \neq y(x)$. We have $\OU(A) \cong \BN(E)^*$ by Proposition \ref{BNDualCharacProp}, so $\OU(A)$ is bounded directed-complete by Lemma \ref{DCLemma} (ii), and has sufficiently many normal states by Proposition \ref{PredualFnlsNormalProp}. But by Proposition \ref{OUDCNormProp}, $\NS(\OU(A)) = \Stat(\OU(A))$. Therefore $\EffS(\NS(\OU(A))) = \EffS(\Stat(\OU(A))) = \EffS(B(\OU(A)^*)) \cong \OU(A)^{**}$ (see \cite[Lemma 4.3]{furber2019} for the fact that $\EffS(B(F)) \cong F^*$ by the restriction map $\rho$, for any base-norm space $F$). So the evaluation mapping $\OU(A) \rightarrow \EffS(\NS(\OU(A)))$ can be considered to be $\rho \circ ev \colon \OU(A) \rightarrow \EffS(\Stat(\OU(A)))$, and since $\rho$ is an isomorphism, it suffices to prove that $\ev \colon \OU(A) \rightarrow \OU(A)^{**}$ is not surjective. 

Using the isomorphisms from Propositions \ref{OUDualCharacProp} and \ref{BNDualCharacProp}, it suffices to show that there is an element of $\BN(A^*)^*$ that is not $\ev(x,\lambda) \circ \lev$ for any $(x,\lambda) \in \OU(A)$. As $A$ is not reflexive, there exists $\phi \in A^{**}$ such that $\phi \neq \ev(x)$ for all $x \in A$, \emph{i.e.} for all $x \in A$ there exists $y \in A^*$ such that $\phi(y) \neq y(x)$. For such a $\phi \in A^{**}$, we can rescale it so that $\|\phi\| \leq \frac{1}{2}$, and we have that for all $(x,\lambda) \in \OU(A)$
\begin{align*}
(\ev(x,\lambda) \circ \lev)(y,0) &= \ev(x,\lambda)(\lev(y,0)) = \lev(y,0)(x,\lambda) = \langle (x,\lambda), (y,0) \rangle \\
 &= y(x) + 0 = y(x) \neq \phi(y) = \langle (y,0), (\phi,\frac{1}{2}) \rangle \\
 &= \lev\left(\phi,\frac{1}{2}\right)(y,0),
\end{align*}
and therefore $\ev(x,\lambda) \circ \lev \neq \lev(\phi,\frac{1}{2})$, so $\lev(\phi,\frac{1}{2})$ is the element of $\BN(A^*)^*$ that we are looking for. 

For part (ii), observe that as $\|\phi\| \leq \frac{1}{2}$, we have $(\phi,\frac{1}{2}) \in \OU(E^{**})_+$, and also $(0,1) - (\phi,\frac{1}{2}) = (-\phi,\frac{1}{2}) \in \OU(E^{**})$, so it is in the effect algebra of $\OU(E^{**})$. Since $B(\BN(E)) \cong \Ball(E) = X$, we have $\EffS(X) \cong \BN(E)^* \cong \OU(A)$, and part (ii) follows from the fullness results \cite[Propositions 3.1 and 3.2]{furber2019}. 
\end{proof}

\end{document}